\documentclass[11pt]{article}

\usepackage{fullpage}

\usepackage{amsmath}
\usepackage{amssymb}
\usepackage{epsfig}
\usepackage{amsthm}
\usepackage{url}
\usepackage{color}
\usepackage{mathtools}

\DeclareMathOperator*{\argmin}{argmin}

\newcommand{\R}{\mathbb{R}}
\newcommand{\abs}[1]{|#1|}

\newcommand{\norm}[1]{\|#1\|}

\newcommand{\dgr}{\mathrm{d}}   
\newcommand{\fs}{(\R^L)^V}      

\newtheorem{theorem}{Theorem}[section]
\newtheorem{lemma}[theorem]{Lemma}

\newtheorem{corollary}[theorem]{Corollary}
\newtheorem{proposition}[theorem]{Proposition}

\newtheorem{definition}[theorem]{Definition}

\begin{document}

\title{Diffusion Methods for Classification with Pairwise
  Relationships}

\author{Pedro F. Felzenszwalb\thanks{Partially supported by NSF under
    grant 1161282 and the Brown-IMPA collaboration program.} \\
  Brown University \\
  Providence, RI, USA \\
  {\tt pff@brown.edu}
  \and Benar F. Svaiter\thanks{Partially supported by CNPq grants
    474996/2013-1, 302962/2011-5 and FAPERJ grant E-26/201.584/2014.} \\
  IMPA \\
  Rio de Janeiro, RJ, Brazil \\
  {\tt benar@impa.br}} \maketitle

\begin{abstract}
  We define two algorithms for propagating information in
  classification problems with pairwise relationships.  The algorithms
  are based on contraction maps and are related to non-linear
  diffusion and random walks on graphs.  The approach is also related
  to message passing algorithms, including belief propagation and mean
  field methods.  The algorithms we describe are guaranteed to
  converge on graphs with arbitrary topology.  Moreover they always
  converge to a unique fixed point, independent of initialization.  We
  prove that the fixed points of the algorithms under consideration
  define lower-bounds on the energy function and the max-marginals of
  a Markov random field.  The theoretical results also illustrate a
  relationship between message passing algorithms and value iteration
  for an infinite horizon Markov decision process.  We illustrate the
  practical application of the algorithms under study with numerical
  experiments in image restoration and stereo depth estimation.
\end{abstract}

\section{Introduction}

In many classification problems there are relationships among a set of
items to be classified.  For example, in image reconstruction problems
adjacent pixels are likely to belong to the same object or image
segment.  This leads to relationships between the labels of different
pixels in an image.  Energy minimization methods based on Markov
random fields (MRF) address these problems in a common framework
\cite{Besag74,WJ08,KF09}.  Within this framework we introduce two new
algorithms for classification with pairwise information.  These
algorithms are based on contraction maps and are related to
non-linear diffusion and random walks on graphs.

The setting under consideration is as follows.  Let $G=(V,E)$ be an
undirected simple graph and $L$ be a set of labels. A labeling of $V$
is a function $x : V \to L$ assigning a label from $L$ to each vertex
in $V$.
Local information is modeled by a cost
$g_i(a)$ for
assigning label $a$ to vertex $i$.  Information on label compatibility  
for neighboring vertices
is modeled by a cost
$h_{ij}(a,b)$ for assigning label $a$ to vertex $i$ and
label $b$ to vertex $j$. 
The cost for a labeling $x$ is defined by an
energy function,
\begin{equation}
  F(x) = \sum_{i \in V} g_i(x_i) + \sum_{\{i,j\} \in E} h_{ij}(x_i,x_j).
\end{equation}
In the context of MRFs the energy function defines a Gibbs
distribution on random variables $X$ associated with the
vertices $V$,
\begin{align}
  p(X=x) & = \frac{1}{Z} \exp(-F(x)).
\end{align}

Minimizing the energy $F(x)$ corresponds to maximizing $p(X=x)$.  This
approach has been applied to a variety of problems in image processing
and computer vision \cite{FZ11}.  A classical example involves
restoring corrupted images \cite{GG84,Besag86}.  In image restoration there is
a grid of pixels and the problem is to estimate an intensity value
for each pixel.  To restore an image $I$ one looks for an image $J$ that
is similar to $I$ and is smooth almost everywhere.  Similarity between
$I$ and $J$ is defined by local costs at each pixel.  The smoothness
constraint is defined by pairwise costs between neighboring pixels in
$J$.

\subsection{Basic Definitions and Overview of Results}

Let $G=(V,E)$ be an undirected, simple, connected graph, with more
than one vertex.  For simplicity let $V=\{1,\dots,n\}$.
Let $N(i)$ and $\dgr(i)$ denote respectively the set of
neighbors and the degree of vertex $i$,
\begin{align*}
  N(i) = \{j \in V \;|\; \{i,j\} \in E\}, \quad \dgr(i)=|N(i)|.
\end{align*}

Let $L$ be a set of labels.  For each vertex $i \in V$ we have a
non-negative cost for assigning label $a$ to vertex $i$, denoted by
$g_i(a)$.  These costs capture local information about the label of
each vertex.  For each edge $\{i,j\} \in E$ we have a non-negative
cost for assigning label $a$ to vertex $i$ and label $b$ to vertex
$j$, denoted equally by $h_{ij}(a,b)$ or $h_{ji}(b,a)$.  These costs capture
relationships between labels of neighboring vertices.
\begin{itemize}
  \item[] $g_i:L\to [0,\infty)$ for $i \in V$;
  \item[] $h_{ij},h_{ji}:L^2 \to [0,\infty)$ for $\{i,j\} \in E$ with
$h_{ij}(a,b) = h_{ji}(b,a)$
\end{itemize}

Let $x \in L^V$ denote a labeling of $V$ with labels from $L$.  A
cost for $x$ that takes into account both local information
at each vertex and the pairwise relationships can be defined by an
energy function $F:L^V\to\R$,
\begin{align}
  \label{eq:F}
  F(x)=\sum_{i\in V} g_i(x_i) +\sum_{\{i,j\}\in E} h_{ij}(x_i,x_j).
\end{align}
This leads to a natural optimization problem where we look for a
labeling $x$ with minimum energy.

Throughout the paper we assume $L$ is finite.  The optimization
problem defined by $F$ is NP-hard even when $|L|=2$ as it can be used
to solve the independent set problem on $G$.  It can also be used to
solve coloring with $k$ colors when $|L|=k$.  The optimization problem
can be solved in polynomial time using dynamic programming when $G$ is
a tree \cite{BB72}.  More generally dynamic programming leads to
polynomial optimization algorithms when the graph $G$ is chordal (triangulated) and has bounded tree-width.  

Min-sum (max-product) belief propagation
\cite{WJ08,KF09} is a local message passing algorithm that is
equivalent to dynamic programming when $G$ is a tree.  Both dynamic
programming and belief propagation aggregate local costs by sequential
propagation of information along the edges in $E$.

For $i \in V$ we define the value function $f_i:L\to\R$,
\begin{align}
  \label{eq:f} 
  f_i(\tau) = \min_{\substack{x \in L^V\\ x_i=\tau}} F(x).
\end{align}
In the context of MRFs the value functions are also known as
\emph{max-marginals}.
The value functions are also what is computed by the dynamic programming and belief propagation algorithms for minimizing $F$ when $G$ is a tree.
Each value function defines a cost for assigning a
label to a vertex that takes into account the whole graph. If $x^*$
minimizes $F(x)$ then $x_i^*$ minimizes $f_i(\tau)$, and when
$f_i(\tau)$ has a unique minimum we can minimize $F(x)$ by selecting
\begin{equation}
x^*_i = \argmin_{\tau} f_i(\tau).
\end{equation}

A local belief is a function $\gamma:L\to\R$.  A
field of beliefs specifies a local belief 
for each vertex in $V$, and is an element of
\begin{align}
  \fs = \{\varphi=(\varphi_1,\ldots,\varphi_N)\;|\; \varphi_i:L\to\R\}.
\end{align}
We define two algorithms in terms of maps,
\begin{align*}
  T:\fs\to\fs,\\
  S:\fs\to\fs.
\end{align*}

The maps $T$ and $S$ are closely related.  Both maps are contractions,
but each of them has its own unique fixed point.  Each of these maps
can be used to define an algorithm to optimize $F(x)$ based on
fixed point iterations and local decisions.

For $z\in\{T,S\}$ we start from an initial field of beliefs
$\varphi^0$ and sequentially compute $$\varphi^{k+1}=z(\varphi^k).$$

Both $S^k(\varphi^0)$ and $T^k(\varphi^0)$ converge to the unique
fixed points of $S$ and $T$ respectively.  After convergence to a
fixed point $\varphi$ (or a bounded number of iterations in practice)
we select a labeling $x$ by selecting the label minimizing the belief
at each vertex (breaking ties arbitrarily),
\begin{align}
  x_i = \argmin_\tau \varphi_i(\tau).
\end{align}

The algorithms we consider depend on parameters $p \in (0,1)$, $q=1-p$
and weights $w_{ij} \in [0,1]$ for each $i \in V$ and $j \in N(i)$.
The weights from each vertex are constrained to sum to one,
\begin{equation}
  \label{eq:wsum}
  \sum_{j \in N(i)} w_{ij} = 1, \qquad \forall i\in V.
\end{equation}
These weights can be interpreted in terms of transition probabilities
for a random walk on $G$.  In a uniform random walk we have $w_{ij} =
1/\dgr(i)$.  Non-uniform weights can be used to capture additional
information about an underlying application.  For example, in the case
of stereo depth estimation (Section~\ref{sec:stereo}) we have used
non-uniform weights that reflect color similarity between neighboring
pixels.  We note, however, that while we may interpret the results of
the fixed point algorithms in terms of transition probabilities in a
random walk, the algorithms we study are deterministic.

The maps $S$ and $T$ we consider are defined as follows,
\begin{definition}
  \label{df:maps}
  \begin{align}
    \label{eq:T}
    &(T \varphi)_i(\tau)
    = p g_i(\tau) + \sum_{j\in N(i)} \min_{u_j\in L}\;
    \dfrac{p}{2}h_{ij}(\tau,u_j)
    +q w_{ji} \varphi_j(u_j) \\
    \label{eq:S}
    &(S \varphi)_i(\tau)
    = p g_i(\tau) + \sum_{j\in N(i)} w_{ij} \min_{u_j\in L}\;
    p h_{ij}(\tau,u_j)
    +q\varphi_j(u_j) 
  \end{align}
\end{definition}

The map defined by $T$ corresponds to a form of non-linear diffusion
of beliefs along the edges of $G$.  The map defined by $S$ corresponds
to value iteration for a Markov decision process (MDP)
\cite{Bertsekas05} defined by random walks on $G$.
We show that both $S$ and $T$ are contractions.  Let $\bar{\varphi}$ be the
fixed point of $T$ and $\hat{\varphi}$ be the fixed point of $S$.
We show $\bar{\varphi}$ defines a lower bound on
the energy function $F$, and that $\hat{\varphi}$ defines lower bounds
on the value functions $f_i$,
\begin{align}
  \sum_{i \in V} \bar{\varphi}_i(x_i) & \leq F(x),\qquad \forall x\in L^V, \\
  \hat{\varphi}_i(\tau) & \le f_i(\tau),\qquad \forall i\in V,\; \tau\in L.
\end{align}

In Section~\ref{sec:T} we study the fixed point iteration algorithm
defined by $T$ and the relationship between $\bar{\varphi}$ and $F$.
To the extent that $\sum_{i \in V} \bar{\varphi}_i(x_i)$
approximates $F(x)$ this justifies  selecting a labeling $x$ by
minimizing $\bar{\varphi}_i$ at each vertex.  This approach is
related to mean field methods and variational inference with
the Gibbs distribution $p(X=x)$ \cite{WJ08,KF09}.

In Section~\ref{sec:S} we study the algorithm defined by $S$ and the
relationship between $\hat{\varphi}_i$ and $f_i$.  To the extent that
$\hat{\varphi}_i(\tau)$ approximates $f_i(\tau)$ this justifies
selecting a labeling $x$ by minimizing $\hat{\varphi}_i$ at each
vertex.  We also show a connection between the fixed point
$\hat{\varphi}$ and optimal policies of a Markov decision process.
The process is defined in terms of random walks on
$G$, with transition probabilities given by the weights $w_{ij}$.

\subsection{Examples}

\begin{figure}[t]
\begin{centering}  
  \includegraphics[height=1in]{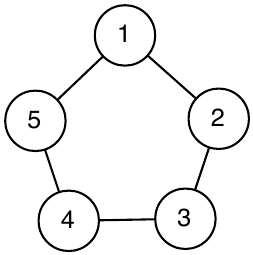}
  \vspace{1cm}
  
  \begin{tabular}{cc}
  \includegraphics[width=3in]{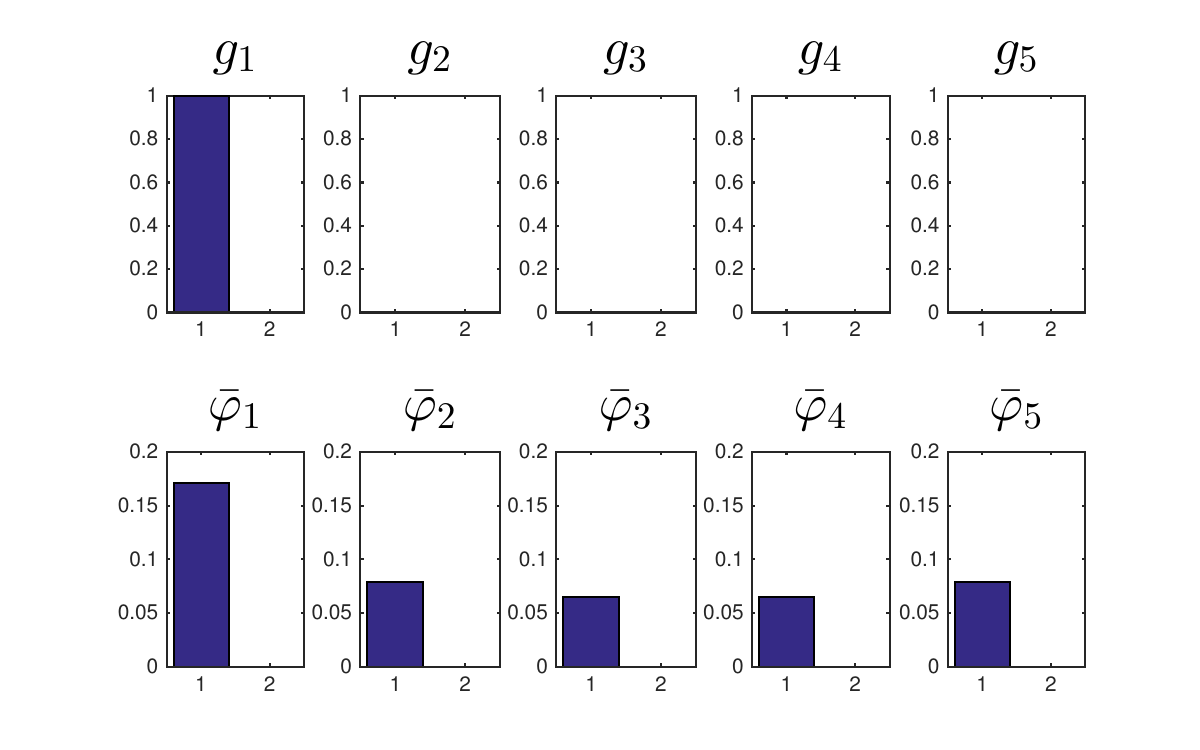} &
  \includegraphics[width=3in]{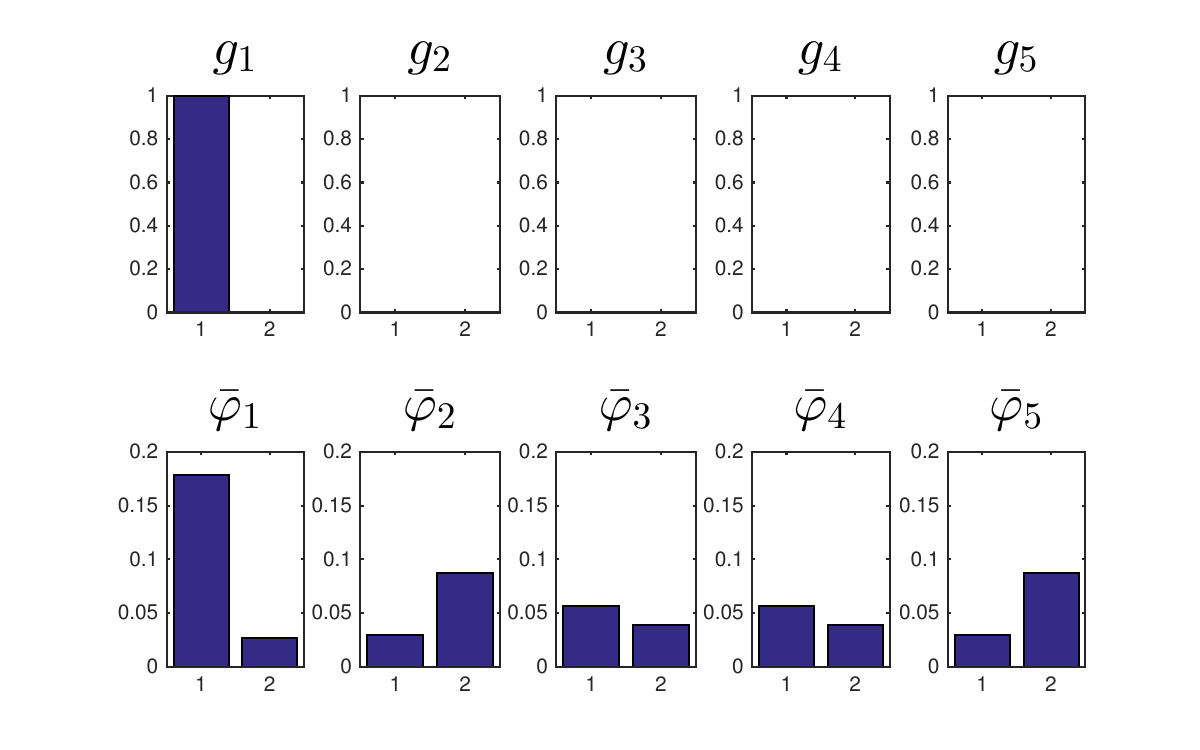} \\
  {\bf (a)} Attractive relationships & {\bf (b)} Repulsive relationships
\end{tabular}
\end{centering}
\caption{The fixed points of $T$ on two problems defined on the graph
  above.  In this case $L=\{1,2\}$.  In both cases the local costs
  $g_i$ are all zero except for vertex 1 who has a preference towards
  label 2.  In {\bf (a)} the pairwise costs encourage neighboring
  vertices to take the same label.  In {\bf (b)} the pairwise costs
  encourage neighboring vertices to take different labels.}
\label{fig:example}
\end{figure}

Figure~\ref{fig:example} shows two examples of fixed points of $T$
when the graph $G=(V,E)$ is a cycle with 5 vertices.  In this case we
have a binary labeling problem $L=\{1,2\}$.  The local costs are
all zero except that vertex 1 has a preference for label 2.  This is
encoded by a cost for label 1,
\begin{align}
  g_1(1) & = 1, \\
  g_1(2) & = 0, \\
  g_i(a) & = 0, \qquad \forall i \neq 1, \; a\in L.
\end{align}
In example (a) we have pairwise costs
that encourage equal labels for neighboring vertices,
\begin{align}
  h_{ij}(a,b) =
  \begin{cases}
    0 \qquad a = b \\
    1 \qquad a \neq b,
  \end{cases}
\end{align}
In example (b) we have pairwise costs that encourage different labels
for neighboring vertices,
\begin{align}
  h_{ij}(a,b) =
  \begin{cases}
    1 \qquad a = b \\
    0 \qquad a \neq b,
  \end{cases}
\end{align}

Figure~\ref{fig:example} shows a graphical representation of the local
costs for each vertex, and the value of $\bar{\varphi}$, the fixed
point of $T$, on each example.  In (a) local selection of $x_i$
minimizing $\bar{\varphi}_i$ leads to $x=(2,2,2,2,2)$.  In (b) local
selection of $x_i$ minimizing $\bar{\varphi}_i$ leads to
$x=(2,1,2,2,1)$.  In both examples the resulting labeling $x$ is the
global minimum of $F(x)$.  For these examples we used $p=0.1$ and
$w_{ij} = 1/\dgr(i)$.

Of course in general local minimization of
$\bar{\varphi}$ does not lead to a labeling minimizing $F(x)$
and it would be interesting to characterize when this happens.

\subsection{Related Work}

For general graphs $G$, when the pairwise costs $h_{ij}(a,b)$ define a
metric over $L$ there are polynomial time approximation algorithms for
the optimization problem defined by $F$ \cite{KT02}.  In some
important cases the optimization problem can be solved using graph
cuts and maximum flow algorithms \cite{GPS89,BVZ01,Boros02,KZ04}.
This includes in particular the case of MAP estimation for an Ising
model with an external field \cite{GPS89}.

The algorithms we study are closely related to message passing
methods, in particular to min-sum (or equivalently max-product) belief
propagation (BP) \cite{WJ08,KF09}. When the graph $G$ is a tree, BP
converges and solves the optimization problem defined by $F$.
Unfortunately BP is not guaranteed to converge and it can
have multiple fixed points for general graphs. Some form of dampening
can help BP converge in practice.
The algorithms we study provide a simple alternative to min-sum belief
propagation that is guaranteed to converge to a unique fixed point,
regardless of initialization. The algorithms are also guaranteed to
converge ``quickly''.

One approach for solving the optimization problem defined by $F$
involves using a linear program (LP) relaxation.  The optimization problem
can be posed using a LP with a large number of constraints and relaxed
to obtain a tractable LP over the \emph{local polytope} \cite{WJW05}.
Several message passing methods have been motivated in terms of this
LP \cite{MGW09}.  There are also recent methods which use message
passing in the inner loop of an algorithm that converges to the
optimal solution of the local polytope LP relaxation 
\cite{RAW10,SSKS12}.
In Section~\ref{sec:lp} we characterize the fixed point of $S$ using
a different LP.

The mean-field algorithm \cite{WJ08,KF09} is an iterative method for
approximating the Gibbs distribution $p(x)$ by a factored
distribution $q(x)$,
\begin{equation}
  q(x) = \prod_{i \in V} q_i(x_i).
\end{equation}
The mean-field approach involves minimization of the KL divergence
between $p$ and $q$ using fixed point iterations that repeatedly
update the factors $q_i$ defining $q$.  A drawback of the approach is
that the fixed point is not unique and the method is sensitive to
initialization.

The algorithm defined by $T$ is related to the mean-field method in
the sense that the fixed points of $T$ appear to approximate $F(x)$ by
a function $H(x)$ that is a sum of local terms,
\begin{equation}
  H(x) = \sum_{i \in V} \bar{\varphi}_i(x_i).
\end{equation}
We do not know, however, if there is a measure under which the
resulting $H(x)$ is an optimal approximation to $F(x)$ within the
class of functions defined by a sum of local terms.

\section{Preliminaries}

The algorithms we study are efficient in the following sense.  Let
$m=|E|$ and $k=|L|$.  Each iteration in the fixed point algorithm
involves evaluating $T$ or $S$.  This can be done in $O(mk^2)$ by
``brute-force'' evaluation of the expressions in
Definition~\ref{df:maps}.  In many applications, including in image
restoration and stereo matching, the pairwise cost $h_{ij}$ has
special structure that allows for faster computation using the
techniques described in \cite{FH12}.  This leads to an $O(mk)$
algorithm for each iteration of the fixed point methods.  Additionally,
the algorithms are easily parallelizable.

The fixed point 
algorithms defined by $T$ and $S$ converge quickly because the
maps are contractions in $\fs$.

Let $z:\R^K \to \R^K$ and $\norm{x}$ be a norm in $\R^K$.  For $\gamma
\in (0,1)$,   $z$ is a $\gamma$-contraction if
\begin{equation}
  \norm{z(x)-z(y)} \le \gamma\norm{x-y}.
\end{equation}
When $z$ is a contraction it has a unique fixed point $\bar{x}$.  It
also follows directly from the contraction property that 
fixed point iteration $x_k = z(x_{k-1})$ converges to $\bar{x}$ quickly,
\begin{equation}
  \norm{x_k-\bar{x}} \le \gamma^k \norm{x_0-\bar{x}}.
\end{equation}

The weights $w_{ij}$ in the definition of $T$ and $S$ define a random
process that generates random walks on $G$.  We have a Markov chain
with state space $V$.  Starting from a vertex $Q_0$ we generate an
infinite sequence of random vertices $(Q_0,Q_1,\ldots)$ with
transition probabilities
\begin{equation}
p(Q_{t+1}=j|Q_t=i) = w_{ij}.
\end{equation}
A natural choice for the weights is $w_{ij} = 1/\dgr(i)$, corresponding
to moving from $i$ to $j$ with uniform probability over $N(i)$.
This choice leads to uniform random
walks on $G$ \cite{Lovasz93}.  

We consider in $\fs$ the partial order
\begin{align}
  \varphi \leq \psi \iff \varphi_i(\tau) \leq \psi_i(\tau) \;\;
  \forall i\in V,\; \forall \tau\in L.
\end{align}
It follows trivially from the definitions of $T$ and $S$ that both
maps preserve order in $\fs$, 
\begin{equation}
  \varphi \leq \psi \Rightarrow T\varphi \leq T\psi,\;
  S\varphi \leq S\psi.
\label{eq:order}
\end{equation}

We claim that for any $\alpha\in\R^V$,
\begin{align}
  \label{eq:fsum}
  \sum_{i\in V}\sum_{j\in N(i)} w_{ji} \alpha_j = \sum_{j\in V} \alpha_j.
\end{align}
This follows from re-ordering the double summation and the constraints that the weights out of each vertex
sum to one,
\begin{align*}
  \sum_{i\in V}\sum_{j\in N(i)} w_{ji} \alpha_j = \sum_{j\in V} \sum_{i \in N(j)} w_{ji} \alpha_j = \sum_{j \in V} \alpha_j
\end{align*}

We note that the algorithms defined by $T$ and $S$ are related in the
following sense.  For a regular graph with degree $d$, if we let
$w_{ij} = 1/d$ the maps $T$ and $S$ are equivalent up to rescaling if
the costs in $T$ and $S$ are rescaled appropriately.

\section{Algorithm defined by $T$ (Diffusion)}
\label{sec:T}

In this section we study the fixed point algorithm defined by $T$.  We
show that $T$ is a contraction in $\fs$ and that the fixed point
of $T$ defines a ``factored'' lower bound on $F$.  

We start by showing that $T$ is a contraction with respect
to the norm on $\fs$ defined by
\begin{align}
  \norm{\varphi}_{\infty,1}=\sum_{i\in V}\norm{\varphi_i}_\infty.
\end{align}

\begin{lemma}
  \label{lm:ctT}
  (Contraction) For any $\varphi,\psi\in \fs$
  \begin{align}
    \label{eq:ct1}
    \norm{(T\varphi)_i - (T\psi)_i}_\infty & \le q \sum_{j \in N(i)}
    w_{ji} \norm{\varphi_j - \psi_j}_\infty \qquad \forall i\in V,\\
    \label{eq:ct2}
    \norm{(T\varphi)-(T\psi)}_{\infty,1} & \leq q
    \norm{\varphi-\psi}_{\infty,1}.
  \end{align}
\end{lemma}

\begin{proof}
  Take $i\in V$ and $\tau \in L$. For any $x\in L^V$ 
  \begin{align*}
    (T\varphi)_i(\tau)
    & = p g_i(\tau)+
    \sum_{j\in N(i)} \min_{u_j \in L} \dfrac{p}{2}h_{ij}(\tau,u_j) +
    q w_{ji} \varphi_j(u_j) \\
    & \leq p g_i(\tau)+
    \sum_{j\in N(i)}\dfrac{p}{2}h_{ij}(\tau,x_j) +
    q w_{ji} \varphi_j(x_j) \\
    & \leq p g_i(\tau) +
    \sum_{j\in N(i)}\dfrac{p}{2}h_{ij}(\tau,x_j) +
    q w_{ji} (\psi_j(x_j) + \abs{\varphi_j(x_j)-\psi_j(x_j)}) \\
    & \leq p g_i(\tau) +
    \sum_{j\in N(i)}\dfrac{p}{2}h_{ij}(\tau,x_j) +
    q w_{ji} (\psi_j(x_j) + \norm{\varphi_j-\psi_j}_\infty) 
  \end{align*}
  Since the inequality defined by the first and last terms above
  holds for any $x$, it holds when $x$
  minimizes the last term.  Therefore
\[
(T\varphi)_i(\tau)\leq (T\psi)_i(\tau) +
q \sum_{j\in N(i)} w_{ji} \norm{\varphi_j-\psi_j}_\infty.
\]
Since this inequality holds interchanging $\varphi$ with $\psi$ we have
\[
\abs{(T\varphi)_i(\tau)-(T\psi)_i(\tau)} \leq
q \sum_{j\in N(i)} w_{ji} \norm{\varphi_j-\psi_j}_\infty.
\]
Taking the $\tau$ maximizing the left hand side proves
(\ref{eq:ct1}).  To prove (\ref{eq:ct2}), we sum the inequalities
(\ref{eq:ct1}) for all $i \in V$ and use (\ref{eq:fsum}).
\end{proof}

The contraction property above implies the fixed point algorithm
defined by $T$ converges to a unique fixed point independent on
initialization.  It also implies the distance to the fixed point
decreases quickly, and we can bound the distance to the fixed point
using either the initial distance to the fixed point or the distance between
consecutive iterates (a readily available measure).

\begin{theorem}
  \label{th:Tconv}
  The map $T$ has a unique fixed point $\bar{\varphi}$ and for
  any $\varphi \in \fs$ and integer $k \ge 0$,
  \begin{align}
    \norm{\bar{\varphi} - T^k \varphi}_{\infty,1}
    & \le q^k \norm{\bar{\varphi}-\varphi}_{\infty,1}, \\
    \norm{\bar{\varphi}-\varphi}_{\infty,1} 
    & \le \dfrac{1}{p} \norm{T\varphi - \varphi}_{\infty,1}.
  \end{align}
\end{theorem}
\begin{proof}
  Existence and uniqueness of the fixed point, as well as the first
  inequality follows trivially from Lemma~\ref{lm:ctT}.  To prove the
  second inequality observe that since $T^k\varphi$ converges to
  $\bar{\varphi}$,
  \begin{equation}
    \norm{\bar{\varphi}-\varphi}_{\infty,1} \le
    \sum_{k=0}^\infty \norm{T^{k+1}\varphi - T^k\varphi}_{\infty,1} \le
    \sum_{k=0}^\infty q^k \norm{T\varphi - \varphi}_{\infty,1}.
  \end{equation}
  Now note that since $p \in (0,1)$ and $p+q=1$,
  \begin{equation}
    \sum_{k=0}^\infty q^k p = 1 \implies
    \sum_{k=0}^\infty q^k = \frac{1}{p}.
  \end{equation}
\end{proof}

The map $T$ and the energy function $F$ are related as follows.

\begin{proposition}
  \label{pr:TF}
  For any $\varphi\in\fs$ and $x\in L^V$
  \begin{align}
  \sum_{i\in V}(T\varphi)_i(x_i)\leq p F(x) + q \sum_{i \in V} \varphi_i(x_i).
  \end{align}
\end{proposition}

\begin{proof}
  Direct use of the definition of $T$ yields
  \begin{align*}
    \sum_{i\in V}(T\varphi)_i(x_i)
    & =   \sum_{i\in V} pg_i(x_i)+\sum_{j\in N(i)} \min_{u_j\in L} \frac{p}{2} h_{ij}(x_i,u_j) + q w_{ji} \varphi_j(u_j) \\
    & \le \sum_{i\in V} pg_i(x_i)+\sum_{j\in N(i)} \frac{p}{2} h_{ij}(x_i,x_j) + q w_{ji} \varphi_j(x_j) \\
    & = p \left( \sum_{i\in V}g_i(x_i) + \sum_{i\in V}\sum_{j\in N(i)} \frac{1}{2} h_{ij}(x_i,x_j) \right)
    + q \sum_{i\in V} \sum_{j \in N(i)} w_{ji} \varphi_j(x_j) \\
    & = p F(x) + q \sum_{j\in V}\varphi_j(x_j),
  \end{align*}
  where the last equality follows from the fact that
  $h_{ij}(x_i,x_j)=h_{ji}(x_j,x_i)$ and Equation~(\ref{eq:fsum}).
\end{proof}

Now we show the fixed point of $T$ defines a lower bound
on $F$ in terms of a sum of local terms.

\begin{theorem}
  \label{th:b1}
  Let $\bar{\varphi}$ be the fixed point of $T$ and
  $$H(x) = \sum_{i \in V} \bar{\varphi}_i(x_i).$$
  Then $0 \leq \bar{\varphi}$ and $H(x) \le F(x)$.
\end{theorem}

\begin{proof}
  The fact that $H(x) \le F(x)$ follows directly from
  Proposition~\ref{pr:TF}.
  
  To prove $0 \le \bar{\varphi}$ consider the sequence
  $(0,T0,T^20,\ldots)$.  The non-negativity of $g_i$ and $h_{ij}$
  implies $0 \le T0$.  Since $T$ is order preserving (\ref{eq:order}) it follows by
  induction that $T^k0 \le T^{k+1}0$ for all $k\ge0$.  Since the
  sequence is pointwise non-decreasing and converges to
  $\bar{\varphi}$ we have $0 \le \bar{\varphi}$.
\end{proof}

Theorem~\ref{th:b1} allows us to compute both a lower and an upper
bound on the optimal value of $F$, together with a solution where $F$ 
attains the upper bound.

\begin{corollary}
  \label{cr:bracket}
  Let $\bar\varphi$ be the fixed point of $T$ and
  \begin{align*}
    \bar x_i = \argmin_{\tau}\bar\varphi_i(\tau) \;\; \forall i\in V,
  \end{align*}
  then for any $x^*$ minimizing $F$,
  \begin{align*}
    \sum_{i\in V}\bar\varphi_i(\bar x_i)\leq F(x^*)\leq F(\bar x).
  \end{align*}
\end{corollary}

\begin{proof}
  If $x^*$ is a minimizer of $F$, then the inequality $F(x^*)\leq
  F(\bar x)$ holds trivially.  We can use the definition of $\bar x$
  to conclude that
  \begin{align*}
    \sum_{i\in V}\bar\varphi_i(\bar x_i) \leq
    \sum_{i\in V}\bar\varphi_i(x^*_i) \leq F(x^*),    
  \end{align*}
  where the second inequality follows from Theorem~\ref{th:b1} 
\end{proof}

\subsection{Linear Programming Formulation}
\label{sec:lp}

Here we provide a LP characterizing for the fixed point of $T$.  We
note that the LP formulation described here is different from the
standard LP relaxation for minimizing $F(x)$ which involves the local
polytope described in \cite{WJW05}.

Consider
the following LP which depends on a vector of coefficients $a$ in $\fs$,
\begin{align*}
  & \max_\varphi a^T \varphi \\
  & \varphi_i(u_i) \le pg_i(u_i) + \sum_{j \in N(i)} \frac{p}{2} h_{ij}(u_i,u_j)
  + q w_{ji} \varphi_j(u_j) & \qquad \forall i\in V, \forall u\in L^V.  
\end{align*}
Note that the constraints in the LP are equivalent to $\varphi \le
T\varphi$.  Next we show that this LP has a unique solution
which equals the fixed point of $T$ whenever every coefficient is positive,
independent of their specific values.  For example, $\bar \varphi$ is the
optimal solution when $a$ is the vector of ones.

\begin{theorem}
  \label{tr:lp}
  If $a$ is a non-negative vector the fixed point of
  $T$ is an optimal solution for the LP.  If $a$ is a positive vector
  the fixed point of $T$ is the unique optimal solution for the LP.
\end{theorem}

\begin{proof}
  Let $\bar{\varphi}$ be the fixed point of $T$.  First note that
  $\bar{\varphi}$ is a feasible solution since
  $\bar{\varphi} \le T\bar{\varphi}$.

  Let $\varphi \in \fs$ be any feasible solution for the LP.  The linear
  constraints are equivalent to $\varphi \le T\varphi$.  Since $T$ preserves
  order it follows by induction that $T^k \varphi \le T^{k+1}
  \varphi$ for all $k \ge 0$.  Since the sequence
  $(\varphi,T\varphi,T^2\varphi,\ldots)$ converges to $\bar{\varphi}$
  and it is pointwise non-decreasing we conclude $\varphi \le
  \bar{\varphi}$.
  
  If $a$ is non-negative we have $a^T\varphi \le a^T \bar{\varphi}$ and
  therefore $\bar{\varphi}$ must be an optimal solution for the LP.
  If $a$ is positive and $\varphi \neq \bar{\varphi}$ we have $a^T\varphi <
  a^T\bar{\varphi}$.  This proves the fixed point is the unique
  optimal solution for the LP.
\end{proof}

\section{Algorithm defined by $S$ (Optimal Control)}
\label{sec:S}

In this section we study the algorithm defined by $S$.  We start
by showing that $S$ corresponds to value iteration for an infinite
horizon discounted Markov decision process (MDP) \cite{Bertsekas05}.  

An infinite horizon discounted MDP is defined by a tuple
$(Q,A,c,t,\gamma)$ where $Q$ is a set of states, $A$ is a set of
actions and $\gamma$ is a discount factor in $\R$.  The cost function
$c:Q \times A \to \R$ specifies a cost $c(s,a)$ for taking action $a$
on state $s$.  The transition probabilities $t:Q \times A \times Q \to
\R$ specify the probability $t(s,a,s')$ of moving to state $s'$ if we
take action $a$ in state $s$.

Let $o$ be an infinite sequence of state and action pairs,
$o=((s_1,a_1),(s_2,a_2),\ldots) \in (Q \times A)^\infty$.  The
(discounted) cost of $o$ is
\begin{equation}
  c(o) = \sum_{k=0}^\infty \gamma^k c(s_k,a_k).
\end{equation}

A policy for the MDP is defined by a map $\pi:Q \rightarrow A$,
specifying an action to be taken at each state.  The value of a state
$s$ under the policy $\pi$ is the expected cost of an infinite
sequence of state and action pairs generated using $\pi$ starting at
$s$,
\begin{equation}
  v_\pi(s) = E[c(o) | \pi, s_1=s].
\end{equation}

An optimal policy $\pi^*$ minimizes $v_\pi(s)$ for every
starting state.  Value iteration computes $v_{\pi^*}$ as the fixed
point of ${\cal L}:\mathbb{R}^Q \to \mathbb{R}^Q$,
\begin{equation}
  ({\cal L} v)(s) = \min_{a \in A} c(s,a) +
  \gamma \sum_{s' \in Q} t(s,a,s') v(s').
\end{equation}
The map ${\cal L}$ is known to be a $\gamma$-contraction \cite{Bertsekas05} with respect
to the $\norm{\cdot}_\infty$ norm.

Now we show that $S$ is equivalent to value iteration for an MDP
defined by random walks on $G$.  Intuitively we have states defined by
a vertex $i \in V$ and a label $a \in L$.  An action involves
selecting a different label for each possible next vertex, and the
next vertex is selected according to a random walk defined by the
weights $w_{ij}$.

\begin{lemma}
  \label{lm:MDP}
  Define an MDP $(Q,A,c,t,\gamma)$ as follows.  The states are pairs
  of vertices and labels $Q = V \times L$. The actions specify a
  label for every possible next vertex $A = L^V$.  The discount factor
  is $\gamma = q$.  The transition probabilities and cost
  function are defined by
  \begin{align}
    t((i,\tau),u,(j,\tau')) & =
    \begin{cases}
      w_{ij} & j \in N(i),\;\tau' = u_j\\
      0 & \text{otherwise}
    \end{cases} \\
    c((i,\tau),u) & =
    p g_i(\tau) + \sum_{j \in N(i)} p w_{ij}h_{ij}(\tau,u_j)
  \end{align}
  The map $S$ is equivalent to value iteration for this MDP.
  That is, if $\varphi_i(\tau) = v((i,\tau))$ then
  $$(S\varphi)_i(\tau) = ({\cal L}v)((i,\tau)).$$
\end{lemma}

\begin{proof}
  The result follows directly from the definition of the MDP, ${\cal
    L}$ and $S$.
  \begin{align}
    ({\cal L} v)((i,\tau))
    & = \min_{u \in L^V} c((i,\tau),u) + \gamma \sum_{(j,\tau') \in Q} t((i,\tau),u,(j,\tau'))v(j,\tau') \\
    & = \min_{u \in L^V} pg_i(\tau) + \sum_{j \in N(i)} pw_{ij}h_{ij}(\tau,u_j) + q \sum_{j \in N(i)} w_{ij} v(j,u_j) \\
    & = pg_i(\tau) + \sum_{j \in N(i)} w_{ij} \min_{u_j \in L} ph_{ij}(\tau,u_j) + q v(j,u_j) \\
    & = (S\varphi)_i(\tau)
\end{align}
\end{proof}

The relationship to value iteration shows $S$ is a contraction and we
have the following results regarding fixed point iterations with $S$.

\begin{theorem}
  The map $S$ has a unique fixed point $\hat \varphi$ and
  for any $\varphi \in \fs$ and integer $k \ge 0$,
  \begin{align}
    \norm{\hat{\varphi}-S^k\varphi}_\infty
    & \leq q^k \norm{\hat{\varphi}-\varphi}_\infty, \\
    \norm{\hat{\varphi}-\varphi}_\infty &
    \leq \frac{1}{p} \norm{S\varphi-\varphi}_\infty.
  \end{align}
\end{theorem}

\begin{proof}
  The first inequality follows directly from Lemma~\ref{lm:MDP} and the
  fact that ${\cal L}$ is a $\gamma$-contraction with $\gamma=q$. The
  proof of the second inequality is similar to the proof of the
  analogous result for the map $T$ in Theorem~\ref{th:Tconv}.
\end{proof}

\subsection{Random Walks}

The formalism of MDPs is quite general, and encompasses the fixed point
algorithm defined by $S$.  In this section we further analyze this fixed
point algorithm and provide an interpretation using one-dimensional problems
defined by random walks on $G$.

The weights $w_{ij}$ define a random process that generates infinite
walks on $G$.  Starting from some vertex in $V$ we repeatedly move to a
neighboring vertex, and the probability of moving from $i \in V$ to $j
\in N(i)$ in one step is given by $w_{ij}$.

An infinite walk $\omega=(\omega_1,\omega_2,\ldots) \in V^\infty$ can
be used to define an energy on an infinite sequence of labels
$z=(z_1,z_2,\ldots) \in L^\infty$,
\begin{equation}
  F_\omega(z) = \sum_{t = 0}^\infty pq^t g_{\omega_t}(z_t) + pq^t h_{\omega_t \omega_{t+1}}(z_t,z_{t+1}).
\end{equation}
The energy $F_\omega(z)$ can be seen as the energy of a pairwise classification
problem on a graph $G'=(V',E')$ that is an infinite path,
\begin{align}
  V'&=\{1,2,\ldots\}, \\
  E'&=\{\{1,2\},\{2,3\},\ldots\}.
\end{align}
The graph $G'$ can be interpreted as a one-dimensional ``unwrapping''
of $G$ along the walk $\omega$.  This unwrapping defines a map from
vertices in the path $G'$ to vertices in $G$.

Consider a policy $\pi : V \times L \times V \to L$ that specifies 
$z_{k+1}$ in terms of $\omega_k$, $z_k$ and $\omega_{k+1}$,
\begin{equation}
  z_{k+1} = \pi(\omega_k,z_k,\omega_{k+1}).
\end{equation}
Now consider the expected value of $F_\omega(z)$ when $\omega$ is a random walk
starting at $i \in V$ and $z$ is a sequence of labels
defined by the policy $\pi$ starting with $z_1=\tau$,
\begin{equation}
  v_\pi(i,\tau) = E[F_\omega(z)|\omega_1=i,z_1=\tau,z_{k+1} = \pi(\omega_k,z_k,\omega_{k+1})].
\end{equation}
There is an optimal policy $\pi^*$ that minimizes $v_\pi(i,\tau)$ for
every $i \in V$ and $\tau \in L$.  Let $\hat{\varphi}$ be the fixed
point of $S$.  Then $\hat{\varphi}_i(\tau) = v_{\pi^*}(i,\tau)$.  This
follows directly from the connection between $S$ and the MDP described
in the last section.

\subsection{Bounding the Value Functions of $F$}

Now we show that $\hat \varphi$ defines lower bounds on the
value functions (max-marginals) $f_i$ defined in (\ref{eq:f}).  We start by showing that $f_i$
can be lower bounded by $f_j$ for $j \in N(i)$.

\begin{proposition}
  \label{pr:flowerbound}
  Let $i \in V$ and $j \in N(i)$.
  \begin{align}
    f_i(u_i) & \ge pg_i(u_i) + \min_{u_j} ph_{ij}(u_i,u_j) + qf_j(u_j), \\
    f_i(u_i) & \ge pg_i(u_i) + \sum_{j \in N(i)} w_{ij} \min_{u_j} ph_{ij}(u_i,u_j) + qf_j(u_j).
  \end{align}
\end{proposition}

\begin{proof}
The second inequality follows from the first one by taking a convex combination over $j \in N(i)$.  To prove the first inequality note that,
  \begin{align}
    f_i(u_i)
    & = \min_{\substack{x \in L^V\\ x_i=u_i}} F(x) \\
    & = \min_{u_j \in L} \min_{\substack{x \in L^V\\ x_i=u_i, x_j=u_j}} F(x) \\
    & = \min_{u_j \in L} \min_{\substack{x \in L^V\\ x_i=u_i, x_j=u_j}} pF(x)+qF(x) \\
    & \ge pg_i(u_i) + \min_{u_j \in L} ph_{ij}(u_i,u_j) + \min_{\substack{x \in L^V\\ x_i=u_i, x_j=u_j}} qF(x) \\
    & \ge pg_i(x_i) + \min_{u_j \in L} ph_{ij}(u_i,u_j) + \min_{\substack{x \in L^V\\ x_j=u_j}} qF(x) \\
    & = pg_i(x_i) + \min_{u_j \in L} ph_{ij}(u_i,u_j) + qf_j(u_j).
  \end{align}
The first inequality above follows from $F(x) \ge g_i(x_i) +
h_{ij}(x_i,x_j)$ since all the terms in $F(x)$ are non-negative.  The
second inequality follows from the fact that we are minimizing $F(x)$
over $x$ with fewer restrictions.
\end{proof}

The map $S$ and the value functions are related as follows.

\begin{proposition}
  \label{pr:Sf}
  Let $f = (f_1,\ldots,f_N) \in \fs$ be a field of beliefs defined by the value functions.
  \begin{align}
    Sf \le f.
  \end{align}
\end{proposition}

\begin{proof}
The result follows directly from Proposition~\ref{pr:flowerbound}.
\end{proof}

Now we show that the fixed point of $S$ defines lower bounds on the value functions.

\begin{theorem}
  Let $\hat{\varphi}$ be the fixed point of $S$.  Then
  $$0 \le \hat{\varphi}_i(\tau) \le f_i(\tau).$$
\end{theorem}

\begin{proof}
  Since the costs $g_i$ and $h_{ij}$ are non-negative we have $0 \le
  S0$.  Using the fact that $S$ preserves order we can conclude $0 \le
  \hat{\varphi}$.

  Since $Sf\leq f$ and $S$ preserves order, $S^kf\leq f$ for all $k$. To end
  the proof, take the limit $k\to \infty$ at the left hand-side of this
  inequality.
\end{proof}

\section{Numerical Experiments}

In this section we illustrate the practical feasibility of the
proposed algorithms with preliminary experiments in computer vision
problems.  

\subsection{Image Restoration}

The goal of image restoration is to estimate a clean image $z$ from a
noisy, or corrupted, version $y$.  A classical approach to solve this
problem involves looking for a piecewise smooth image $x$ that is
similar to $y$ \cite{GG84,BZ87}.  In the weak membrane model
\cite{BZ87} the local costs $g_i(a)$ penalize differences between $x$
and $y$ while the pairwise costs $h_{ij}(a,b)$ penalize differences
between neighboring pixels in $x$.  In this setting, the graph
$G=(V,E)$ is a grid in which the vertices $V$ correspond to pixels and
the edges $E$ connect neighboring pixels.  The labels $L$ are possible
pixel values and a labeling $x$ defines an image.  For our experiments
we use $L=\{0,\ldots,255\}$ corresponding to the possible values in an
8-bit image.

To restore $y$ we define the energy $F(x)$ using
\begin{align}
  g_i(x_i) &= (y_i-x_i)^2; \\
  h_{ij}(x_i,x_j) &= \lambda \min((x_i-x_j)^2,\tau).
\end{align}
The local cost $g_i(x_i)$ encourages $x_i$ to be similar to $y_i$.
The pairwise costs depend on two parameters
$\lambda,\tau \in \R$.  The cost $h_{ij}(x_i,x_j)$ encourages $x_i$ to
be similar to $x_j$ but also allows for large differences since the
cost is bounded by $\tau$.  The value of $\lambda$ controls the
relative weight of the local and pairwise costs.  Small values of
$\lambda$ lead to images $x$ that are very similar to the noisy image
$y$, while large values of $\lambda$ lead to images $x$ that are
smoother.

Figure~\ref{fig:restore} shows an example result of image restoration
using the algorithm defined by $T$.  The example illustrates the
algorithm is able to recover a clean image that is smooth almost
everywhere while at the same time preserving sharp discontinuities at
the boundaries of objects.  For comparison we also show the results of
belief propagation.  In this example the noisy image $y$ was obtained
from a clean image $z$ by adding independent noise to each pixel using
a Gaussian distribution with standard deviation $\sigma=20$.  The
input image has 122 by 179 pixels.  We used $\lambda = 0.05$ and $\tau
= 100$ to define the pairwise costs.  For the algorithm defined by $T$
we used uniform weights, $w_{ij} = 1/\dgr(i)$ and $p = 0.001$.  Both
the algorithms defined by $T$ and belief propagation were run for 100
iterations.  We based our implementations on the belief propagation
code from \cite{FH06}, which provides efficient methods for handling
truncated quadratic discontinuity costs.  The algorithm defined by $T$
took 16 seconds on a 1.6Ghz Intel Core i5 laptop computer while
belief propagation took 18 seconds.

\begin{figure}
  \centering
  \begin{tabular}{cccc}
  \includegraphics[width=1.5in]{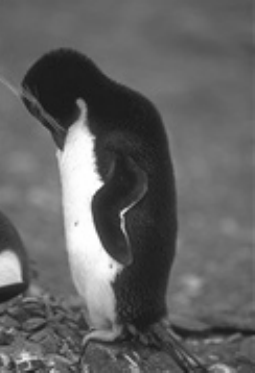} &
  \includegraphics[width=1.5in]{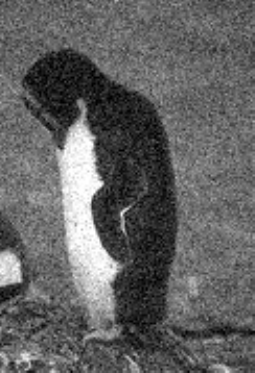} &
  \includegraphics[width=1.5in]{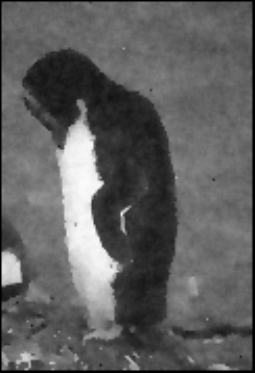} &
  \includegraphics[width=1.5in]{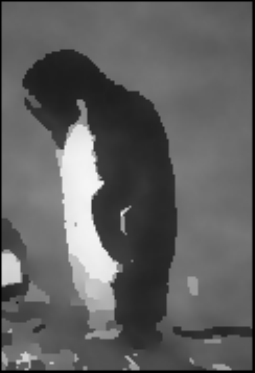} \\
  Original Image & Noisy Image & Output of $T$ & Output of BP \\
  & & RMS error = 8.9 & RMS error = 10.7 \\
  & & Energy = 1519837 & Energy = 650296
  \end{tabular}
  \caption{Image restoration using the fixed point algorithm defined
    by $T$ and BP.  The algorithms were run for 100 iterations.}
  \label{fig:restore}
\end{figure}

The goal of restoration is to recover a clean image $z$.  We evaluate
the restored image $x$ by computing the root mean squared error (RMSE)
between $x$ and $z$.  We see in Figure~\ref{fig:restore} that when
$\lambda=0.05$ and $\tau = 100$ the result of $T$ has lower RMSE value
compared to the result of BP, even though the result of $T$ has
significantly higher energy.  We also evaluate the results of $T$, $S$
and BP using different values of $\lambda$ in
Table~\ref{tb:restoration}.  For all of these experiments we used
$\tau = 100$ and ran each algorithm for 100 iterations.  The minimum
RMSE obtained by $T$ and $S$ is lower than the minimum RMSE obtained
by BP considering different values for $\lambda$, even though $T$ and
$S$ alwayd find solutions that have higher energy compared to BP.  This
suggests the algorithms we propose do a good job aggregating local information
using pairwise constraints, but the energy minimization problem defined
by $F(x)$ may not be the ideal formulation of the restoration problem.

\begin{table}
\centering
  \begin{tabular}{|l|r|r|r|r|r|r|}
  \hline
  & \multicolumn{2}{|c|}{T}
  & \multicolumn{2}{|c|}{S}
  & \multicolumn{2}{|c|}{BP} \\
  \hline
  $\lambda$ & Energy & RMSE & Energy & RMSE & Energy & RMSE \\
  \hline
  0.01 & 659210  & 10.1 & 842459  & 9.1  & 211646 & 15.8 \\
  \hline
0.02 & 943508  & 9.0  & 1220572 & 8.7  & 337785 & 13.4 \\
  \hline
0.05 & 1519837 & 8.9  & 1873560 & 10.9 & 650296 & 10.7 \\
  \hline
0.10 & 2089415 & 11.0 & 2506230 & 14.8 & 1080976 & 10.1 \\
  \hline
0.20 & 2700193 & 14.8 & 2942392 & 17.7 & 1730132 & 12.9 \\
  \hline
\end{tabular}
\caption{Results of restoration using $T$, $S$ and belief propagation
  (BP).  The goal of restoration is to recover the original image $z$.
  We show the energy of the restored image $x$ and the root mean
  squared error (RMSE) between $x$ and $z$.  We show the results of
  the different algorithms for different values of the parameter
  $\lambda$.  Both $T$ and $S$ obtain lower RMSE compared to BP even
  though BP generally obtains results with significantly lower
  energy.}
\label{tb:restoration}
\end{table}

\subsection{Stereo Depth Estimation}
\label{sec:stereo}

In stereo matching we have two images $I_l$ and $I_r$ taken at the
same time from different viewpoints.  Most pixels in one image have a
corresponding pixel in the other image, being the
projection of the same three-dimensional point.  The difference in the
coordinates of corresponding pixels is called the disparity.
We assume the images are rectified such that a pixel $(x,y)$ in $I_l$
matches a pixel $(x-d,y)$ in $I_r$ with $d \ge 0$.  For rectified
images the distance of a three-dimensional point to the 
image plane is inversely proportional to the disparity.

In practice we consider the problem of labeling every pixel in $I_l$
with an integer disparity in $L=\{0,\ldots,D\}$.  In this case a
labeling $x$ is a disparity map for $I_l$.  The
local costs $g_i(a)$ encourage pixels in $I_l$ to be matched to pixels
of similar color in $I_r$.  The pairwise costs $h_{ij}(a,b)$ encourage
piecewise smooth disparity maps.

The model we used in our stereo experiment is defined by,
\begin{align}
  g_{i}(a) &= \min(\gamma, ||I_l(i)-I_r(i-(a,0))||_1); \\
  h_{ij}(a,b) &=
  \begin{cases}
    0 \qquad a = b, \\
    \alpha \qquad |a-b| = 1, \\
    \beta \qquad |a-b| > 1.
  \end{cases}
\end{align}
Here $I_l(i)$ is the value of pixel $i$ in $I_l$ while $I_r(i-(a,0))$
is the value of the corresponding pixel in $I_r$ assuming a disparity
$a$ for $i$.  The $\ell_1$ norm $||I_l(i)-I_r(i-(a,0))||_1$ defines a
distance between RGB values (matching pixels should have similar
color).  The color distance is truncated by $\gamma$ to allow for some
large color differences which occur due to specular reflections and
occlusions.  The pairwise costs depend on two parameter $\alpha, \beta
\in \mathbb{R}$ with $\alpha < \beta$.  The pairwise costs encourage
the disparity neighboring pixels to be similar or differ by 1 (to
allow for slanted surfaces), but also allows for large discontinuities
which occur at object boundaries.

Figure~\ref{fig:stereo} shows an example result of disparity
estimation using the fixed point algorithm defined by $S$.  In this
example we used non-uniform weights $w_{ij}$ to emphasize the
relationships between neighboring pixels of similar color, since those
pixels are most likely to belong to the same object/surface.
The parameters we used for the results in
Figure~\ref{fig:stereo} were defined by,
\begin{align}
  w_{ij} \propto 0.01 + e^{-0.2 ||I_l(i)-I_l(j)||_1},
\end{align}
$p = 0.0001$, $\alpha = 500$, $\beta = 1000$ and $\gamma = 20$.  The
input image has 384 by 288 pixels and the maximum disparity is $D=15$.
The fixed point algorithm was run for 1,000 iterations which took 13
seconds on a laptop computer.

\begin{figure}
  \centering
  \begin{tabular}{cc}
  \includegraphics[width=2.5in]{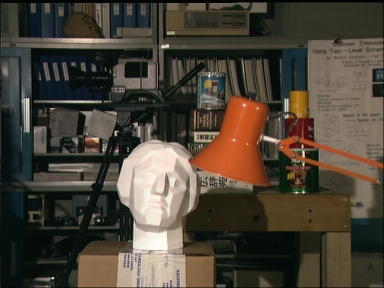} &
  \includegraphics[width=2.5in]{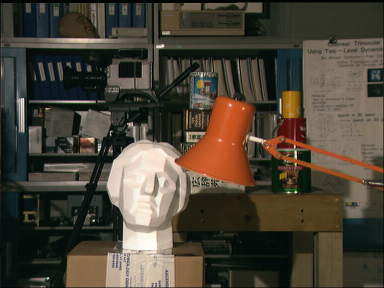} \\
  $I_l$ & $I_r$ \\ \\
  \includegraphics[width=2.5in]{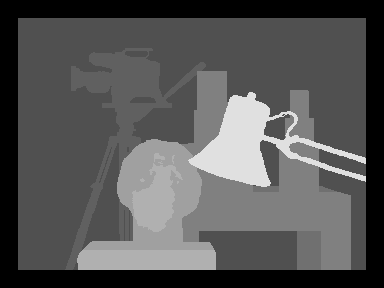} &
  \includegraphics[width=2.5in]{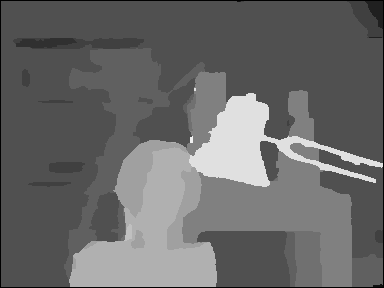} \\
  Ground truth & Result of $S$ \\
  \end{tabular}
  \caption{Stereo disparity estimation using the fixed point algorithm
    defined by $S$ on the Tsukuba image pair.  The algorithm was run
    for 1,000 iterations.}
  \label{fig:stereo}
\end{figure}

We note that the results in
Figure~\ref{fig:stereo} are similar to results obtained min-sum
belief propagation shown in \cite{FH06}.

\section{Conclusion and Future Work}

The experimental results in the last section illustrate the practical
feasibility of the algorithms under study.  Our theoretical results prove these
algorithms are guaranteed to converge to unique fixed points on graphs
with arbitrary topology and with arbitrary pairwise relationships.  This
includes the case of repulsive interactions which often leads to convergence
problems for message passing methods.

Our results can be extended to other contraction maps
similar to $T$ and $S$ and alternative methods for
computing the fixed points of these maps.  Some specific directions
for future work are as follows.

\begin{enumerate}

\item \emph{Asynchronous updates}.  It is possible to define
  algorithms that update the beliefs of a single vertex at a time
  in any order.  As long as all vertices are updated infinitely many
  times, the resulting algorithms converge to the same fixed point as
  the parallel update methods examined in this work.  We conjecture
  that in a \emph{sequential} computation, the sequential update of
  vertices in a ``sweep'' would converge faster than a
  ``parallel'' update.  Moreover, after a sequential update of all vertices,
  the neighbors of those vertices with greater change
  should be the first ones to be updated in the next ``sweep''.

\item \emph{Non-backtracking random walks}.  The algorithms defined by
  $S$ and $T$ can be understood in terms of random walks on $G$.  It
  is possible to define alternative algorithms based on
  non-backtracking random walks.  In particular, starting with the MDP
  in Section~\ref{sec:S} we can increase the state-space $Q$ to keep track
  of the last vertex visited in the walk and define transition
  probabilities that avoid the previous vertex when selecting the next
  one.  The resulting value iteration algorithm becomes very similar
  to belief propagation and other message passing methods that involve
  messages defined on the edges of $G$.
\end{enumerate}


\bibliography{prop}
\bibliographystyle{plain}

\end{document}